\documentclass[runningheads,a4paper]{llncs}

\usepackage{amssymb}
\usepackage{graphicx}
\usepackage{multirow}

\usepackage{url}
\urldef{\mailsa}\path|{marius.nicolae, rajasek}@engr.uconn.edu|    
\newcommand{\keywords}[1]{\par\addvspace\baselineskip
\noindent\keywordname\enspace\ignorespaces#1}

\newcommand{\commentOut}[1]{}

\begin{document}

\mainmatter  

\title{Efficient Sequential and Parallel Algorithms for Planted Motif Search}


%
%
\author{Marius Nicolae
\and Sanguthevar Rajasekaran}
\authorrunning{Efficient Sequential and Parallel Algorithms for Planted Motif Search}

\institute{Dept. of Computer Science and Engineering, Univ. of
Connecticut, Storrs, CT, USA\\
\mailsa}
%
%

\maketitle

\begin{abstract}
\vspace{-0.2in}
Motif searching is an important step in the detection of rare events
occurring in a set of DNA or protein sequences. One formulation of the problem 
is known as $(l,d)$-motif search or Planted Motif Search (PMS). In PMS we are given two 
integers $l$ and $d$ and $n$ biological sequences. We want to find all sequences of length $l$ that
appear in each of the input sequences with at most $d$ mismatches.
The PMS problem is NP-complete. PMS algorithms are typically evaluated on certain 
instances considered challenging. 
This paper presents an exact parallel PMS
algorithm called PMS8. PMS8 is the first algorithm to solve the challenging $(l,d)$ instances $(25,10)$ and
$(26,11)$. PMS8 is also efficient on instances with larger $l$ and $d$ 
such as $(50,21)$. This paper also introduces necessary and sufficient conditions for 3 $l$-mers
to have a common $d$-neighbor.

\vspace{-0.1in}
\end{abstract}

\section{Introduction}
This paper presents an efficient exact parallel algorithm for
the Planted Motif Search (PMS) problem also known as the $(l,d)$ motif problem
\cite{Pev00}.
A string of length $l$ is caller an $l$-mer. 
The number of positions where two $l$-mers $u$ and $v$ differ is called their
Hamming distance and is denoted by $Hd(u,v)$. For any string $T$, $T[i..j]$ is
the substring of $T$ starting at position $i$ and ending at position $j$.  
The PMS problem is the following. Given $n$ sequences $S_1, S_2, \ldots, S_n$ of
length $m$ each, from an alphabet $\Sigma$ and two integers $l$ and $d$, identify
all $l$-mers $M, M \in \Sigma^l$, that occur in at least one location in each
of the $n$ sequences with a Hamming distance of at most $d$. More formally, $M$ is a
motif if and only if $\forall i,1\leq i\leq n, \exists j_i, 1\leq
j_i\leq m-l+1$, such that $Hd(M,S_i[j_i..j_i+m-1])\leq d$.

The PMS problem is essentially the same as the Closest Substring
problem. These problems have applications in PCR primer
design, genetic probe design, discovering potential drug targets, antisense
drug design, finding unbiased consensus of a protein family, creating
diagnostic probes and motif finding (see e.g., \cite{LLM+99}). Therefore,
efficient algorithms for solving the PMS problem are very important in biology
and bioinformatics.

A PMS algorithm that finds all the motifs for a given input is called an exact
algorithm. All known exact algorithms have an exponential worst case runtime
because the PMS problem is NP-complete \cite{LLM+99}.
An exact algorithm can be built using two approaches. One is
sample driven: for all $(m - l +1)^n$ possible 
combinations of $l$-mers coming from different strings, generate the common neighborhood.
   The other is pattern-driven:
for all $\Sigma^l$ possible $l$-mers check which are motifs.
Many algorithms employ a combination of these two techniques.
For example, \cite{YHZG12} and \cite{DBR07} generate the common neighbors for every pair
of $l$-mers coming from two of the input strings. Every neighbor is then matches against
the remaining $n-2$ input strings to confirm or reject it as a motif. Other algorithms (\cite{DRD12,HJG09})
consider groups of three $l$-mers  instead of two.

PMS algorithms are typically tested on instances generated as follows (also
see \cite{Pev00,DBR07}): 20 DNA strings of length 600 are generated
according to the i.i.d model. A random $l$-mer is chosen as a motif and
planted at random location in each input strings. Every planted instance is 
modified in at most $d$ positions. 
For a given integer $l$, the instance $(l,d)$ is defined to be challenging if $d$ is the smallest integer for which the expected number of motifs
of length $l$ that occur in the input by random chance is $\geq 1$. Some of the challenging instances are 
$(13,4),(15,5),(17,6),(19,7),(21,8),(23,9),(25,10),(26,11)$, etc. 

The largest challenging instance solved up to now has been $(23,9)$. To the best of our knowledge
the only algorithm to solve $(23,9)$ has been qPMS7 \cite{DRD12}. 
The algorithm in \cite{DeM11} can solve instances with relatively large $l$
(up to $48$) provided that $d$ is at most $l/4$. However, most of the well known
challenging instances have $d>l/4$. 
PairMotif \cite{YHZG12} can solve instances with larger $l$, such as $(27,9)$ or $(30,9)$, but these are 
significantly less challenging than $(23,9)$.

In this paper we propose a new exact algorithm, PMS8, which can solve both
instances with large $l$ and instances with large $d$. One of the basic steps employed in many
PMS algorithms (such as PMSprune, PMS5, PMS6, and qPMS7) is that of computing all the common neighbors of
three $l$-mers. In qPMS7, this problem is solved using an Integer Linear Programming (ILP) formulation. In particular, a
large number of ILP instances are solved as a part of a preprocessing step and a table is populated. This table is then repeatedly looked up to identify common neighbors of three $l$-mers. This preprocessing step takes a considerable amount of time and the lookup table calls for a large amount of memory. In this paper we offer a novel algorithm for computing all the common neighbors of three $l$-mers. This algorithm eliminates the preprocessing step. In particular, we don't solve any ILP instance. We also don't employ any lookup tables and hence we reduce the memory usage. We feel that this algorithm will find independent applications.
Specifically, we state and prove necessary and sufficient
conditions for $3$ $l$-mers to have a common neighbor (section
\ref{sec_pruning}). 

\vspace{-0.1in}
\section{Methods}
For any $l$-mer $u$ we define its $d$-neighborhood as the set of $l$-mers $v$ for which $Hd(u,v)\leq d$.
For any set of $l$-mers $T$ we define the common $d$-neighborhood of $T$ as the intersection of the 
$d$-neighborhoods of all $l$-mers in $T$. To compute common neighborhoods, a natural approach is to traverse 
the tree of all possible $l$-mers and identify the common neighbors. A pseudocode is given in appendix 
\ref{secGenNeighborhood}. A node at depth $k$, which represents a $k$-mer, is not explored 
deeper if certain pruning conditions are met. Thus, the better the pruning conditions are, the faster 
will be the algorithm. We discuss pruning conditions in section \ref{sec_pruning}. 

PMS8 consists of a sample driven part followed by a pattern
driven part.  In the sample driven part we generate tuples of $l$-mers originating from different strings.
In the pattern driven part we generate the common $d$-neighborhood of such
tuples.    Initially we build a matrix $R$ of
size $n \times (m-l+1$) where row $i$ contains all the $l$-mers in $S_i$. We pick an
$l$-mer $x$ from row 1 of $R$ and push it on a stack. We filter out any
$l$-mer in $R$ at a distance greater than $2d$ from $x$.
Then we pick an $l$-mer from the second row of $R$ and push it on the stack. We filter out any $l$-mer in $R$ that
does not have a common neighbor with the $l$-mers on the stack; then we repeat
the process.  A necessary and sufficient condition for 3 $l$-mers to have a common neighbor is
discussed in section \ref{sec_pruning}. For 4 or more
$l$-mers we only have necessary conditions, so we may generate tuples
that will not lead to solutions. 
If any row becomes empty, we discard the top of the stack, revert to the previous
instance of $R$ and try a different $l$-mer.  If the stack size is above a certain threshold 
(see section \ref{sec_threshold}) we generate the common $d$-neighborhood of the $l$-mers 
on the stack. For each neighbor $M$ we check whether there is at least one
$l$-mer $u$ in each row of $R$ such that $Hd(M,u)\leq d$. If this is true then $M$ is a motif.
PMS8 is illustrated in figure \ref{figFiltering} and its pseudocode is given in appendix \ref{secPMS8}.

\begin{figure}[!t]
\includegraphics[width=\linewidth]{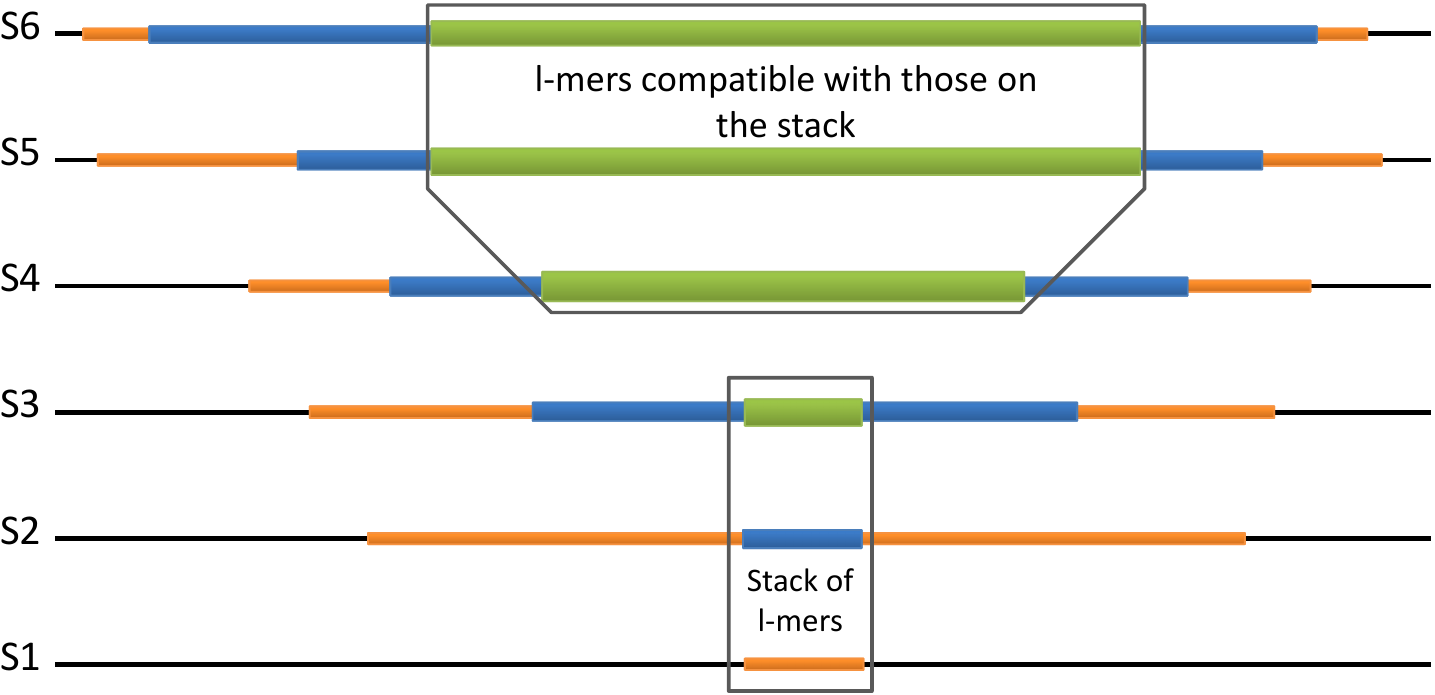}
\caption{\label{figFiltering}
Illustration of PMS8. We repeatedly push $l$-mers from different strings on a stack. From the remaining strings, 
we filter out $l$-mers incompatible with those on the stack. Once the stack has a certain size we generate the 
common $d$-neighbors of its $l$-mers. We compare the neighbors against the surviving $l$-mers to identify the motifs.}
\vspace{-0.1in}
\end{figure}

\vspace{-0.1in}
\subsection{Speedup techniques}
{\bf Sort rows by size.}
An important speedup technique is to reorder the
rows of $R$ by size after every filtering step.  This reduces the number of tuples that we
consider at lower stack sizes. These tuples require the most expensive filtering
because as the stack size increases fewer $l$-mers remain to be filtered.

{\bf Compress $l$-mers.}
\label{secCompressLmers}
We can speed up Hamming distance operations by compressing all the l-mers of $R$ in advance.
For example, for DNA we store 8 characters in a 16 bit integer, divided into 8 groups of 2 bits. 
For every $16$ bit integer $i$ we store in a table the number of non-zero groups of bits in $i$.
To compute the Hamming distance between two $l$-mers we first
perform an exclusive or of their compressed representations. Equal characters produce bits
of 0, different characters produce non-zero bits. Therefore, one table lookup provides the 
Hamming distance for 8 characters. One compressed $l$-mer
requires $l * \lceil \log|\Sigma| \rceil$ bits of storage. However, we only need 
the first $16$ bits of this representation
because the next 16 bits are the same as the first 16 bits of the $l$-mer $8$
positions to the right of the current one. Therefore, the table of
compressed $l$-mers only requires $O(n(m-l+1))$ words of memory.

{\bf Preprocess distances for pairs of $l$-mers.}
\label{secDistPairs}
The filtering step tests many times if two $l$-mers have a
distance of no more than $2d$. Thus, for every pair
of $l$-mers we compute this bit of information in advance. 

{\bf Cache locality.}
We can update $R$ in an efficient manner as follows. Every row in the
updated matrix $R'$ is a subset of the corresponding row in the current matrix
$R$ and thus we can store it in the same memory locations as $R$ by rearranging 
the row elements and keeping track how many of them belong to $R'$. 
This both reduces the memory requirement and improves cache locality:
the surviving $l$-mers in one filtering step will soon be accessed in the next one. 

\subsection{Memory and Runtime }
Since we store all matrices $R$ in the space of a single matrix they
only require $O(n(m-l+1))$ words of memory to which we add $O(n^2)$ words to store
 row sizes. The bits of information for compatible $l$-mer pairs take
$O((n(m-l+1))^2/w)$ words, where $w$ is the number of bits in a machine word.
The table of compressed $l$-mers takes $O(n(m-l+1))$ words. Therefore, the
total memory used by the algorithm is $O(n(n+m-l+1) + (n(m-l+1))^2 / w)$.

\label{sec_threshold}
The more time we spend in the sample driven part, 
the less time we have to spend in the pattern driven part and vice-versa. Ideally we want to
choose the threshold where we switch between the two parts such that their runtimes are almost equal. 
The optimal threshold can be determined empirically by running the algorithm on a small subset of the 
tuples. In practice, PMS8 heuristically estimates the threshold $t$ such that it increases with $d$ and $|\Sigma|$ to avoid
generating very large neighborhoods and it decreases with $m$ to
avoid spending too much time on filtering. A strong closed form 
runtime for the algorithm is difficult to derive. 
A more in depth analysis can be found in appendix \ref{secThreshold}. 
All the results reported in this paper 
have been obtained using the default threshold estimation.

\subsection{Parallel implementation}
To parallelize PMS8 we create $m-l+1$ sub problems,
one for each $l$-mer in the first string. The first string in each sub problem
is an $l$-mer of the original first string and the rest of the strings are the
same as in the original input. 
The processor with rank 0 is a scheduler and the others are workers. 
The scheduler spawns a separate worker thread 
to avoid using one processor just for scheduling. 
The scheduler reads the input and broadcasts it to all workers. 
Then each worker requests a sub problem from the
scheduler, solves it and repeats. The scheduler loops until all 
jobs have been requested and all workers have been notified that no more jobs are available.
At the end, all processors send their motifs to the scheduler which outputs
them. The process is illustrated in figure \ref{figWorkers}.  

\begin{figure}
\vspace{-0.1in}
\begin{minipage}[b]{0.5\linewidth}
\centering
\includegraphics[width=\linewidth]{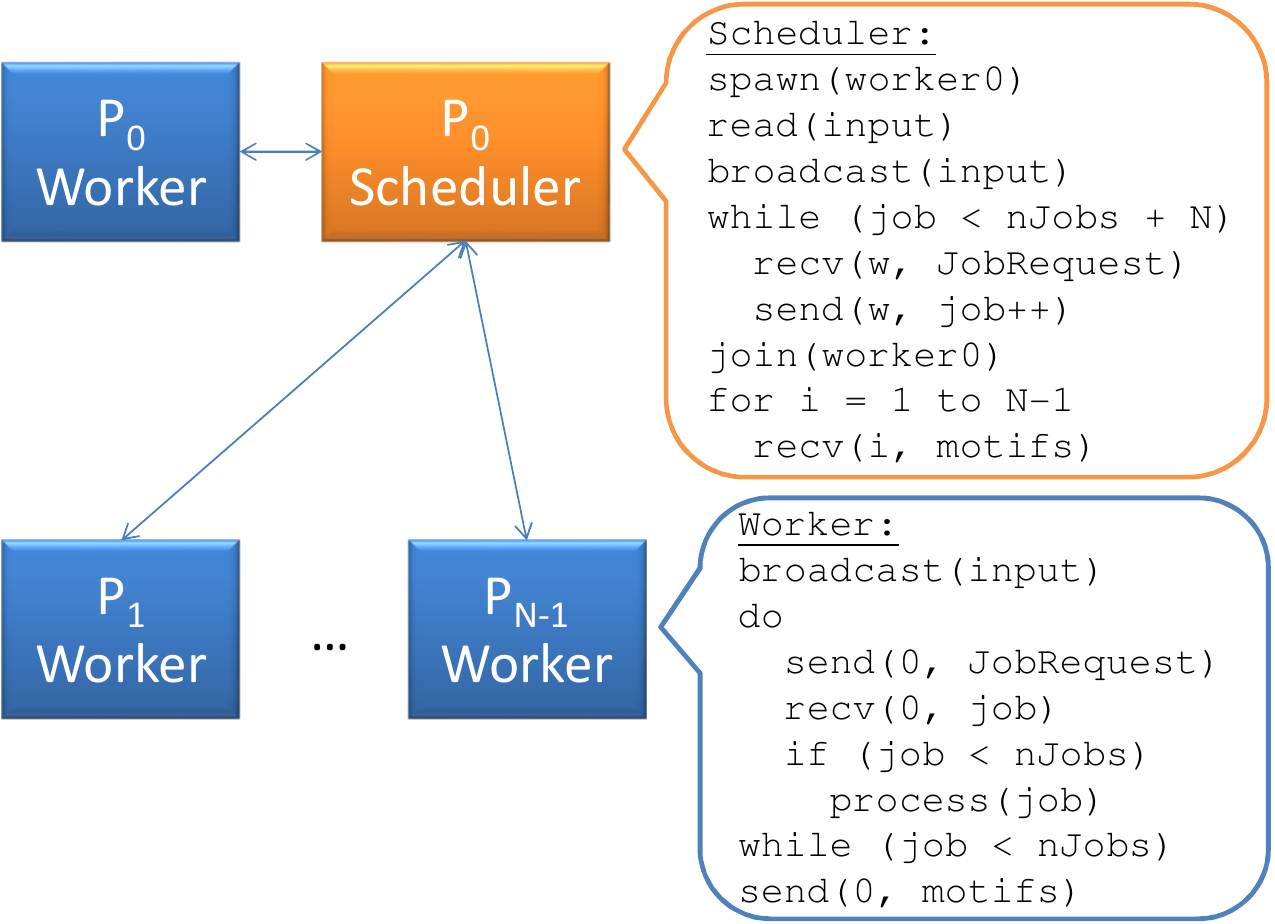}
\caption{\label{figWorkers}Parallel implementation using MPI. Processor 0 is a
scheduler and the others are workers. The scheduler also spawns a
separate thread and uses it as worker.}
\end{minipage}
\hspace{0.3cm}
\begin{minipage}[b]{0.5\linewidth}
\centering
\includegraphics[width=\linewidth]{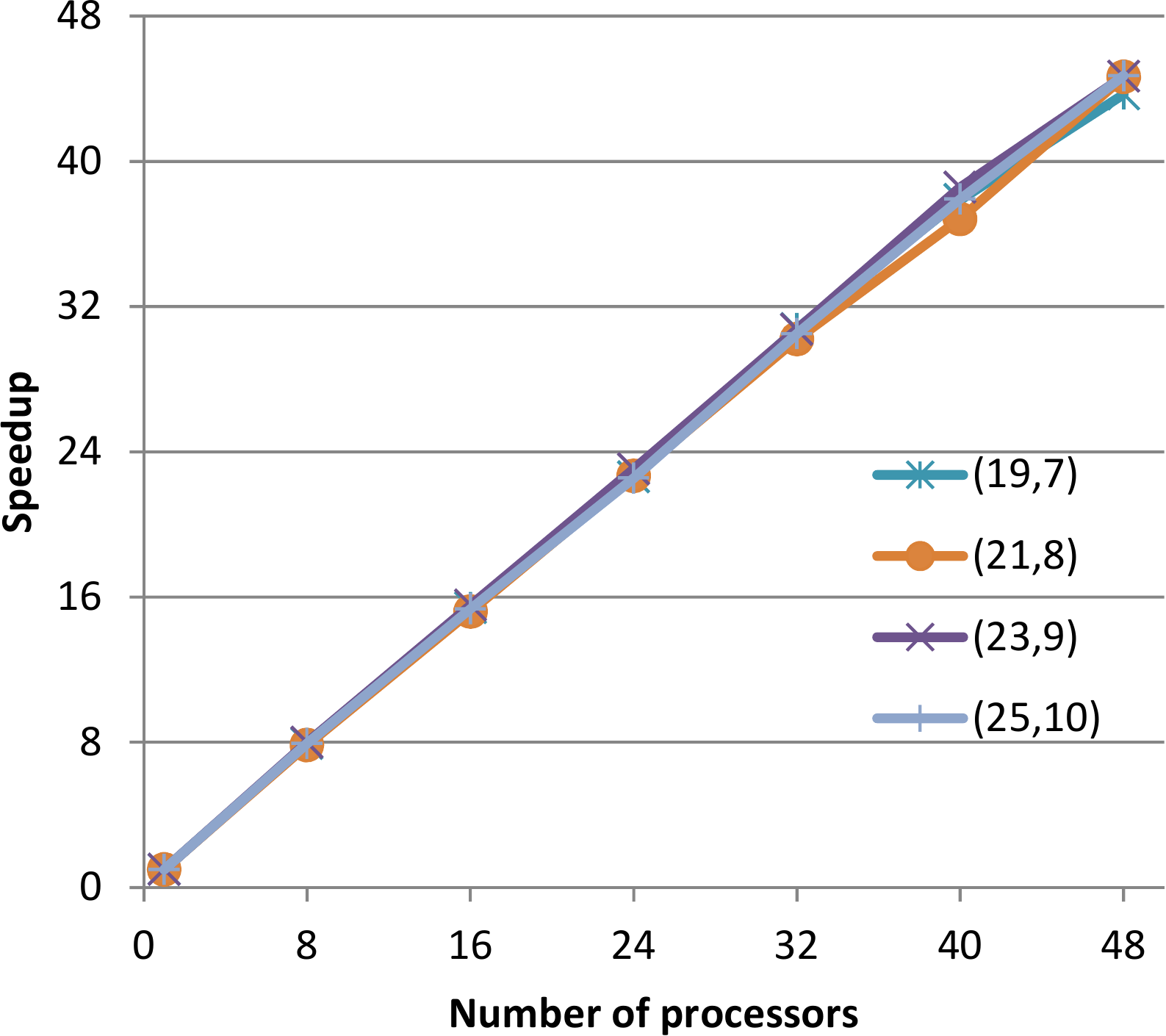}
\caption{\label{figSpeedup}Speedup of the multi-core version of PMS8 over the
single core version, for several datasets.}
\end{minipage}
\vspace{-0.2in}
\end{figure}

\subsection{Pruning conditions}
\label{sec_pruning}

In this section we present pruning conditions applied
for filtering $l$-mers in the sample driven part and for
pruning enumeration trees in the pattern driven part.

Two $l$-mers $a$ and $b$ have a common neighbor $M$ such that $Hd(a,M)\leq d_a$
and $Hd(b,M)\leq d_b$ if and only if $Hd(a,b)\leq d_a+d_b$.
For $3$ $l$-mers, no trivial necessary and sufficient conditions have been
known up to now. In \cite{DRK11} sufficient conditions for 3 $l$-mers
are obtained from a preprocessed table. However, as $l$ increases the memory 
requirement of the table becomes a bottleneck. We will give simple 
necessary and sufficient conditions for 3 $l$-mers to have
a common neighbor. These conditions are also necessary for
more than $3$ $l$-mers. 

Let $T$ be a set of $l$-mers and $M$ be an $l$-mer. If  $\sum_{u\in T}Hd(M,u) > |T| d$  
 then, by the pigeonhole principle, one $l$-mer must have a distance from $M$ greater than $d$.
Therefore, $M$ cannot be a common neighbor of the $l$-mers in $T$.
If we have a lower bound on  $\sum_{u\in T}Hd(M,u)$ for any $M$, then we
can use it as a pruning condition. If the lower bound is greater
than $|T|d$ then there is no common neighbor for $T$.
One such lower bound is the {\em consensus total distance}.

\begin{definition}
Let $T$ be a set of $l$-mers,  where $k=|T|$. For every $i$, the set $T_1[i],T_2[i],..,T_k[i]$
is called the $i$-th column of $T$.
Let $m_i$ be the maximum frequency of any character in column $i$.  Then $Cd(T)=\sum_{i=1..l}k-m_i$ is called the
consensus total distance of $T$.
\end{definition}

The consensus total distance is a lower bound for the total distance
between any $l$-mer $M$ and the $l$-mers in $T$ because, regardless of
$M$, the distance contributed by column $i$ to the total distance is at
least $k-m_i$.
The consensus total distance for a set of two $l$-mers $A$ and $B$ will be
denoted by $Cd(A,B)$.  Also notice that $Cd(A,B)=Hd(A,B)$.
We can easily prove the following lemma.

\begin{lemma}
\label{lemma_consensus}
Let $T$ be a set of $l$-mers and $k=|T|$. Let $d_1,d_2,\ldots d_k$ be
non-negative integers. 
There exists a $l$-mer $M$ such that $Hd(M,T_i)\leq d_i, \forall i$, only if
$Cd(T)\leq \Sigma_{i=1}^kd_i$.
\end{lemma}

\begin{theorem}
\label{th_3sufficient}
Let $T$ be a set of 3 $l$-mers and $d_1,d_2,d_3$ be
non-negative integers. There exists a $l$-mer $M$ such that $Hd(M,T_i)\leq d_i, \forall
i, 1\leq i\leq 3$ if and only if the following conditions hold:
\begin{description}
\item[i)]  $Cd(T_i,T_j)\leq d_i+d_j, \forall i,j, 1 \leq i < j \leq 3$
\item[ii)] $Cd(T) \leq d_1+d_2+d_3$
\end{description}
\end{theorem}

\begin{proof}
The ``only if'' part follows from lemma \ref{lemma_consensus}. For the ``if'' part we show how to
construct a common neighbor $M$ provided that the conditions hold.
 
We say that a column $k$ where
$T_1[k]=T_2[k]=T_3[k]$ is of type $N_0$. If $T_1[k]\neq T_2[k] = T_3[k]$ then
the column is of type $N_1$. If $T_1[k]=T_3[k]\neq T_2[k]$ the
column is of type $N_2$ and if $T_1[k]=T_2[k]\neq T_3[k]$ then the column is of
type $N_3$. If all three characters in the column are distinct, the
column is of type $N_4$. Let $n_i, \forall i, 0\leq i \leq 4$ be
the number of columns of type $N_i$. Consider two cases:

Case 1) There exists $i, 1 \leq i \leq 3$ for which $n_i \geq d_i$. We construct
$M$ as illustrated in the left panel of figure \ref{figProofCase1}. Pick $d_i$ columns of type
$n_i$. For each chosen column $k$ set $M[k]=T_j[k]$ where $j\neq i$. For all
other columns  set $M[k]=T_i[k]$. Therefore $Cd(T_i, M) = d_i$. For $j\neq i$ we know that $Cd(T_i,T_j)
\leq d_i+d_j$ from our assumptions. We also know that
$Cd(T_i,M)+Cd(M,T_j)\leq Cd(T_i,T_j)$ from the triangle inequality.
It follows that $Cd(M,T_j)\leq d_j$. Since $Cd(M,T_j)=Hd(M,T_j)$ it means that
$M$ is indeed a common neighbor of the three $l$-mers.

Case 2) For all $i, 1 \leq i \leq 3$ we have $n_i < d_i$. We construct $M$ as
shown in the right panel of figure \ref{figProofCase2}. For columns $k$ of type $N_0,N_2$ and $N_3$
we set $M[k]=T_1[k]$. For columns of type $N_1$ we set $M[k]=T_2[k]$. 
For any $i,1\leq i \leq 3$ the following applies. If $n_i+n_4\leq d_i$ then the
Hamming distance between $M$ and $T_i$ is less than $d_i$ regardless of what
characters we choose for $M$ in the columns of type $N_4$.
On the other hand, if $n_i+n_4 > d_i$ then $M$ and $T_i$ have to match in at
least $n_i+n_4-d_i$ columns of type $N_4$. Thus, we pick
$max(0,n_i+n_4-d_i)$ columns of type $N_4$ and for each such column $k$ we set
$M[k]=T_i[k]$.
Now we prove that we actually have enough columns to make the
above choices, in other words $\Sigma_{i=1}^3max(0,n_i+n_4-d_i)\leq n_4$. This is equivalent to the following
conditions being true:
\begin{description}
\item[a)] For any $i, 1\leq i \leq 3$ we want $n_i+n_4-d_i \leq n_4$. This is
true because $n_i < d_i$.
\item[b)] For any $i,j, 1\leq i < j \leq 3$ we want
$(n_i+n_4-d_i)+(n_j+n_4-d_j) \leq n_4$. This can be rewritten as
$n_i+n_j+n_4\leq d_i+d_j$. The left hand side is $Hd(T_i, T_j)$ which
we know is less or equal to $d_i+d_j$.
\item[c)] We want $\Sigma_{i=1}^3n_i+n_4-d_i \leq n_4$. This can be rewritten as
$n_1 + n_2 + n_3 + 2 n_4 \leq d_1+d_2+d_3$. The left hand side is $Cd(T)$ which
we know is less than $d_1+d_2+d_3$.
\end{description}
\end{proof}

\begin{figure}
\begin{minipage}[b]{0.5\linewidth}
\centering
\includegraphics[width=\linewidth]{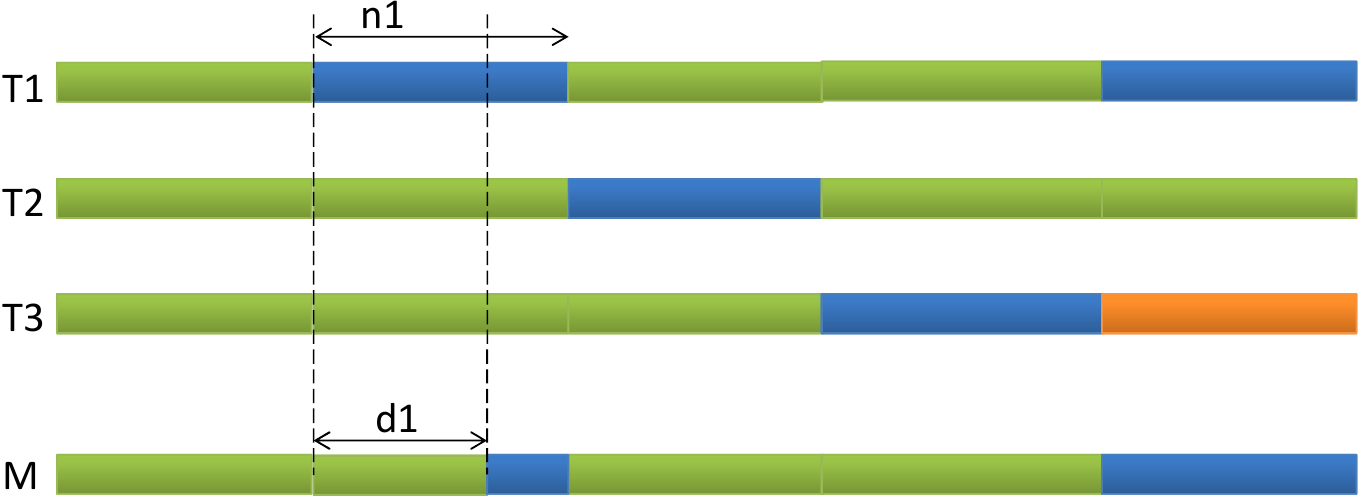}
\end{minipage}
\hspace{0.3cm}
\begin{minipage}[b]{0.5\linewidth}
\centering
\includegraphics[width=\linewidth]{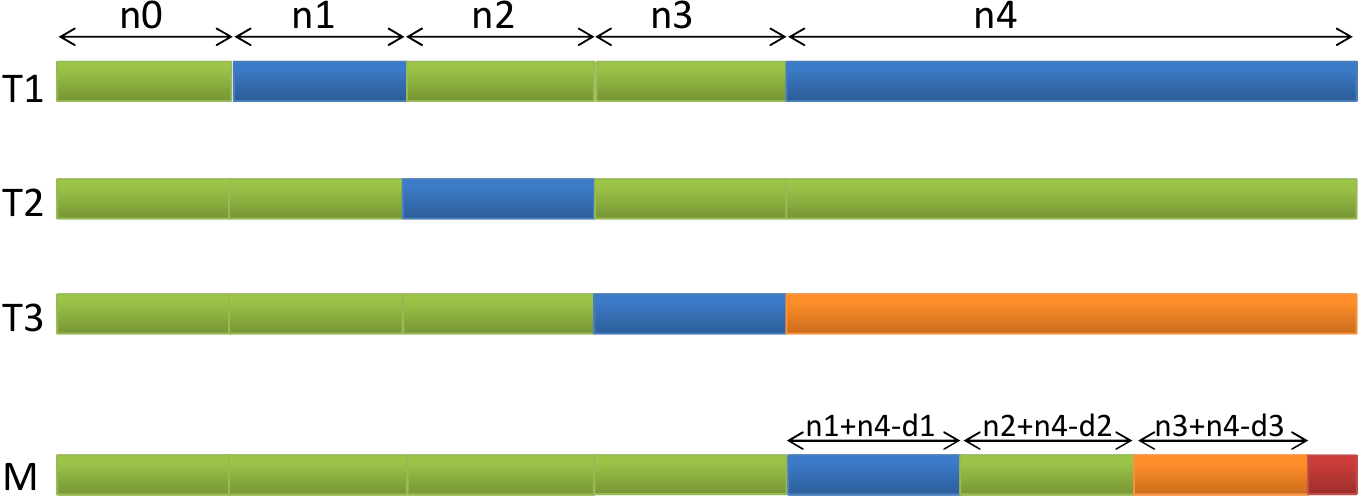}
\end{minipage}
\caption{\label{figProofCase2}\label{figProofCase1}Proof of theorem \ref{th_3sufficient}.
 Case 1 - left figure:  There exists $i, 1 \leq i \leq 3$ for which $n_i\geq d_i$. 
Without loss of generality we assume $i=1$. 
Case 2 - right figure: $n_i<d_i$ for all $i$, $1\leq i \leq 3$. 
The top 3 rows represent the input $l$-mers. The last row shows a common neighbor $M$. In any
column, identical colors represents matches, different colors represent mismatches. 
}
\vspace{-0.1in} \end{figure} \section{Results and Discussion} PMS8 is implemented in C++ and uses OpenMPI for communication between processors. PMS8 was evaluated on the Hornet cluster in the Booth Engineering Center for Advanced Technology (BECAT) at University of Connecticut. The Hornet cluster consists of 64 nodes, each
 equipped with 12 Intel Xeon X5650 Westmere cores 
 and 48 GB of RAM. The nodes use Infiniband
networking for MPI. In our experiments we employed at most 48 cores on at most 4
nodes.

We generated random $(l,d)$ instances according to \cite{Pev00} and as described in the introduction.
For every $(l,d)$ combination we report the average runtime over 5 random instances.
For several challenging instances, in figure \ref{figSpeedup} we present the speedup 
obtained by the parallel version over the single core version. For $p=48$ cores
the speedup is close to $S=45$ and thus the efficiency is $E=S/p=94\%$. 

The runtime of PMS8 on instances with $l$ up to $50$ and $d$ up
to $21$ is shown in figure \ref{figLDTable}. Instances which are expected to have more than $500$ motifs
simply by random chance (spurious motifs) are excluded. The expected number
of spurious motifs was computed as described in appendix \ref{secChallenging}. 
Instances where $d$ is small relative to $l$ are solved efficiently using a single CPU core.
For more challenging instances we report the time taken using 48 cores.

\begin{figure*}
\includegraphics[width=1.1\linewidth]{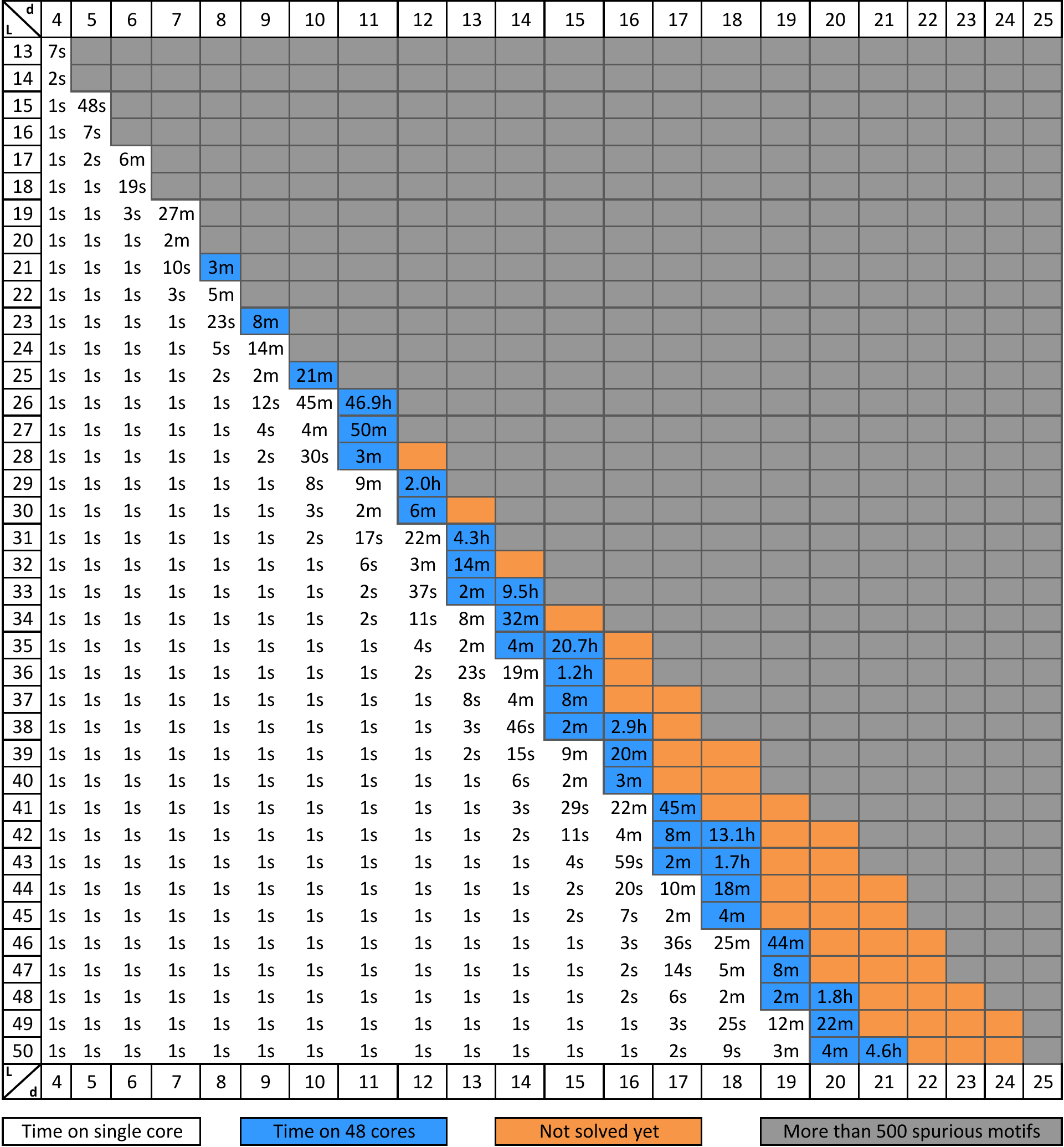}
\caption{\label{figLDTable}PMS8 runtimes for datasets with $l$ up to 50 and $d$ up to 25 
averaged over 5 random datasets.  White background signifies single
core execution.  Blue background signifies execution using 48 cores.
Instances in gray have more than 500 spurious motifs. Orange
cells indicate unsolved instances. Time is reported in seconds (s), minutes (m) or hours (h).}
\end{figure*}

A comparison between PMS8 and qPMS7 \cite{DRD12} on
challenging instances is shown in figure \ref{figCompChallenging}. 
Both programs have been executed on the Hornet cluster. qPMS7 is a sequential algorithm.
PMS8 was evaluated using up to 48 cores.
 The speedup of PMS8 single core over qPMS7 
is shown in figure \ref{figSpeedupQPMS7}. The speedup is high for small instances because
qPMS7 has to load an ILP table. For larger instances the speedup of PMS8 sharply increases. This is expected
because qPMS7 always generates neighborhoods for tuples of $3$ $l$-mers, which become very
large as $l$ and $d$ grow. On the other hand, PMS8 increases the number of $l$-mers in the tuple with the
instance size. With each $l$-mer added to the tuple, the size of
the neighborhood reduces exponentially, whereas the number of neighborhoods
generated increases by a linear factor. The ILP table precomputation requires solving many ILP formulations.
The table then makes qPMS7 less memory efficient than PMS8.
 The peak memory used by qPMS7 for the challenging instances in figure \ref{figCompChallenging} 
 was 607 MB whereas for PMS8 it was 122 MB.  PMS8 is the first algorithm to 
solve the challenging instances (25,10) and (26,11).

\begin{figure}
\begin{minipage}[b]{0.55\linewidth}
\centering
\includegraphics[width=\linewidth]{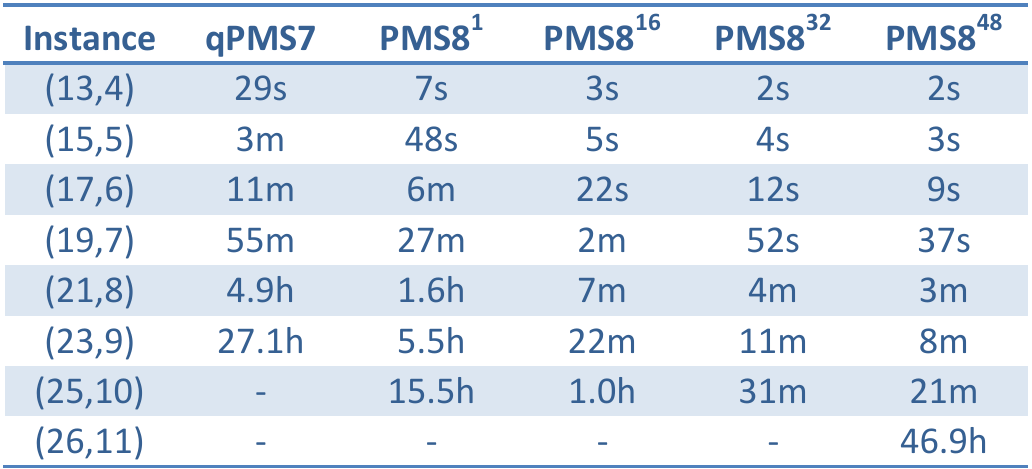}
\caption{\label{figCompChallenging}
Comparison between qPMS7 and PMS8 on
challenging instances. PMS$8^P$ means	 PMS8 used $P$ CPU cores.}
\end{minipage}
\hspace{0.5cm}
\begin{minipage}[b]{0.45\linewidth}
\centering
\includegraphics[width=\linewidth]{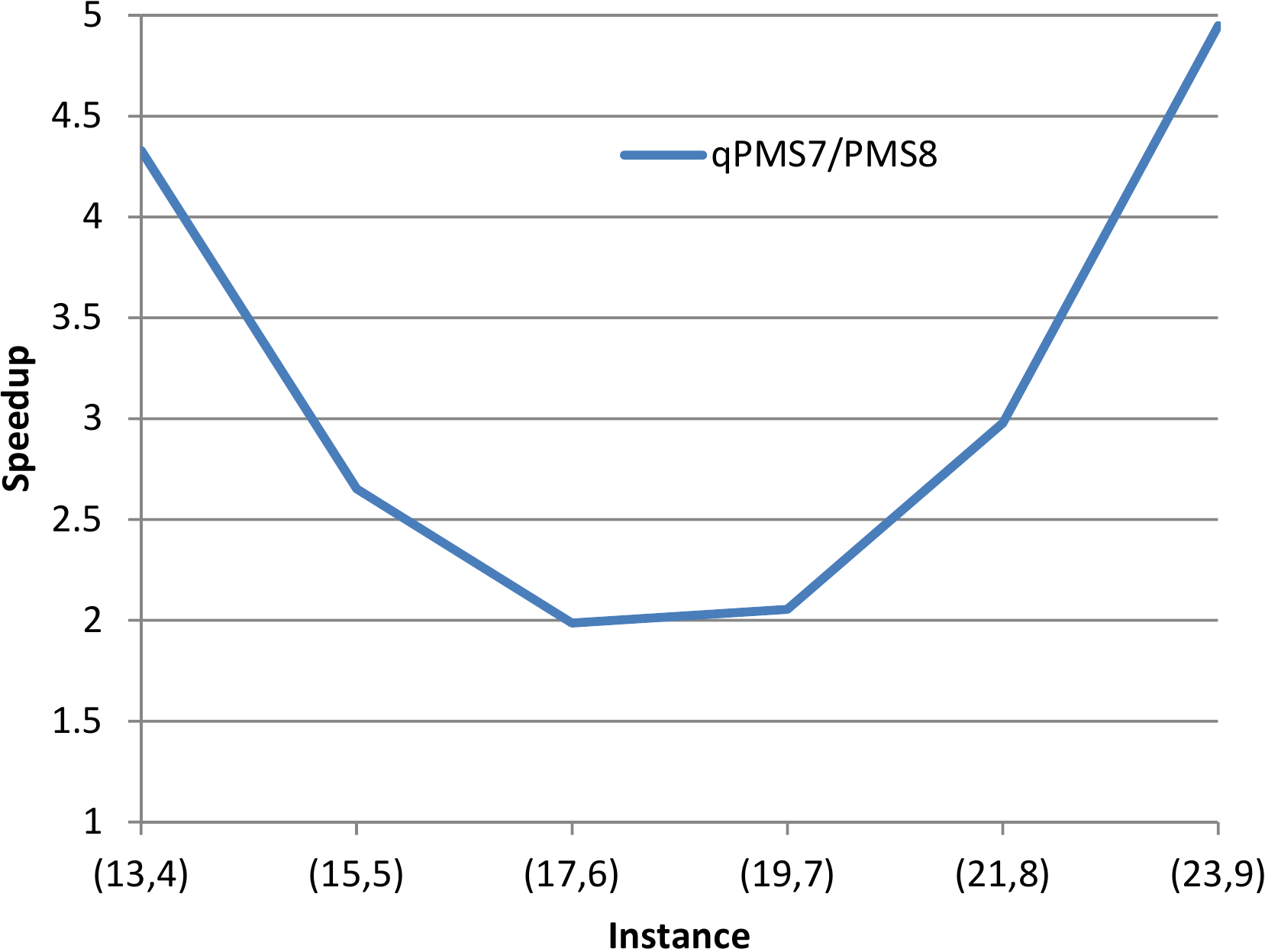}
\caption{\label{figSpeedupQPMS7}
Speedup of PMS8 single core over qPMS7. 
}
\end{minipage}
\end{figure}

Some recent results in the literature have focused on instances other than
the challenging ones presented above. A summary of these results and a comparison
with PMS8 is presented in table \ref{table_comparison}. These
results have been obtained on various types of hardware: single core,
multi-core, GPU, grid. In the comparison, we try to match the number of processors whenever possible.
However, the speed difference is large enough that the hardware is unlikely to play an
important part.

\begin{table*}
\def\arraystretch{1.1}
\begin{tabular}{| p{2.3in} | p{0.5in} | p{0.5in} | p{0.6in} | p{0.4in} | p{0.4in} |}
\hline
Previous algorithm & Instance & Time & Cores & PMS8 Time & PMS8 Cores\\
\hline
Yu et al. 2012 \cite{YHZG12}, PairMotif & (27, 9) & 10h & 1 & 4s & 1\\
\hline
\multirow{2}{*}{Desaraju and Mukkamala 2011 \cite{DeM11}} & (24,6) & 347s & 1 & 1s & 1\\
\cline{2-6}
                                         & (48,12)& 188s & 1 & 1s & 1\\
\hline
\multirow{2}{*}{\vbox{Dasari et al. 2011 \cite{DRZ11}, mSPELLER / gSPELLER}} & (21,8) & 3.7h & 16 & 7m & 16\\
\cline{2-6}
                                                   & (21,8) & 2.2h & \multirow{2}{*}{\vbox{4 GPUs x 240 cores}} & 7m & 16\\
\cline{1-3}\cline{5-6}
Dasari et al. 2010 \cite{DDZ10}, BitBased & (21,8) & 1.1h & & 7m & 16\\
\hline
Dasari and Desh 2010 \cite{DD10}, BitBased & (21,8) & 6.9h & 16 & 7m & 16\\
\hline
Sahoo et al. 2011 \cite{SSRP11} & (16,4) & 106s & 4 & 1s & 1\\
\hline
Sun et al. 2011 \cite{SLHTR11}, TreeMotif &(40,14) & 6h & 1 & 6s & 1\\
\hline
He et al. 2010 \cite{HLWR10}, ListMotif & (40,14) & 28,087s & 1 & 6s &1\\
\hline
Faheem 2010 \cite{Fah10}, skip-Brute Force & (15,4) & 2934s & 96 nodes & 1s & 1\\
\hline
\multirow{3}{*}{Ho et al. 2009 \cite{HJG09}, iTriplet} & (24,8) & 4h & 1 & 5s & 1\\
\cline{2-6}
                                      & (38,12) & 1h & 1 & 1s & 1\\
\cline{2-6}
                                      & (40,12) & 5m & 1 & 1s & 1\\
\hline
\end{tabular}
\caption{\label{table_comparison}Side by side comparison between previous
results in the literature and PMS8. Time is reported in seconds (s), minutes (m) 
or hours (h).  }
\end{table*}

\commentOut{
\section{Conclusions}
Motif searching is an important problem in computer science and bioinformatics. 
We have presented PMS8, an exact parallel motif finding algorithm. PMS8 is faster 
than previous methods for both short and long motifs. PMS8 is the first algorithm 
to solve the challenging instances (25,10) and (26,11).

In section \ref{sec_pruning} we have introduced necessary and sufficient conditions 
for three $l$-mers to have a common neighbor. These novel conditions 
may find independent applications.
}

\section{Acknowledgment}
The authors would like to thank Prof. Chun-Hsi (Vincent) Huang, Dr. Hieu Dinh,
Tian Mi and Gabriel Sebastian Ilie for helpful discussions. This work has been supported in part by the following grants: NSF 0829916 and NIH R01-LM010101.

\bibliographystyle{splncs03}
\bibliography{references}
%
%
\newpage

\appendix
\section{Generating neighborhoods}
\label{secGenNeighborhood}
\begin{quote}
\begin{tabbing}
aa \= aa \= aa \= aa \= aa \= \kill
{\bf Algorithm 1. GenerateNeighborhood}($T,d$)\\
\> {\bf for} ($i=1..|T|$) {\bf do} $r_i:=d;$\\
\> {\bf GenerateNeighborhood}($T,r,1$)\\
\\
{\bf GenerateNeighborhood}($T,r,p$)\\
\>{\bf if} ($p \leq l$) {\bf then} \\
\>\>{\bf if} ({\bf not} prune($T,r$)) {\bf then} \\
\>\>\> {\bf for} $\alpha \in \Sigma$ {\bf do} \\
\>\>\>\> $x_p:=\alpha$ \\
\>\>\>\> {\bf for} ($i=1..|T|$) {\bf do} \\ 
\>\>\>\>\> $T_i':=T_i[2..|S_i|]$\\
\>\>\>\>\> $r_i':=r_i$;\\
\>\>\>\>\> {\bf if} ($T_i[0]\neq \alpha$) {\bf then } $r_i':=r_i'-1;$\\
\>\>\>\> {\bf end for}\\
\>\>\>\> {\bf GenerateNeighborhood}($T', r', p+1$) \\
\>\>\> {\bf end for} \\
\>\>{\bf end if} \\
\>{\bf else}\\
\>\> report $l$-mer $x$ \\
\>{\bf end if}
\end{tabbing}
\end{quote}

\section{PMS8 pseudocode}
\label{secPMS8}
\begin{quote}
\begin{tabbing}
aa \= aa \= aa \= aa \= aa \= \kill
{\bf Algorithm 2. PMS8}($T,d$)\\
\> {\bf for} ($i=1..n$) {\bf do}  $R_i=\{u | u \in S_i\}$ \\
\> $stack=\{\}$ \\
\> {\bf GenerateMotifs}($1, stack, R$) \\
\\
{\bf GenerateMotifs}($p, stack, R$)\\
\>{\bf for} ($u \in R_p$) {\bf do}\\
\>\> stack.push($u$) \\
\>\> $R':=$filter($R$, stack)\\
\>\> {\bf if} ($R'.$size $> 0$) {\bf then}\\
\>\>\> {\bf if} (ThresholdCondition) {\bf then} \\
\>\>\>\>$N:=${\bf GenerateNeighborhood}(stack,$d$)\\
\>\>\>\>{\bf for} ($m \in N$) {\bf do}\\
\>\>\>\>\> {\bf if} (isMotif($m,R'$)) {\bf then} output $m$;\\
\>\>\> {\bf else} \\
\>\>\>\> {\bf GenerateMotifs}($p+1,R'$)\\
\>\> stack.pop()\\
\> {\bf end for}\\
\end{tabbing}
\end{quote}

\section{Challenging instances}
\label{secChallenging}
For a fixed $l$, as $d$ increases, the instance becomes more challenging.
However, as $d$ increases, the number of false positives also increases, because
many motifs will appear simply by random chance.  The expected number of
spurious motifs in a random instance can be estimated as follows (see e.g., \cite{DBR07}). The number 
of $l$-mers in the neighborhood of a given $l$-mer $M$ is $N(\Sigma, l,
d)=\Sigma_{i=0}^d(_d^l)(|\Sigma|-1)^d$. The probability that $M$ is a
$d$-neighbor of a random $l$-mer is $p(\Sigma, l, d)=N(\Sigma, l
,d)/|\Sigma|^l$. The probability that $M$ has at least one $d$-neighbor among
the $l$-mers of a string of length $m$ is thus
$q(m,\Sigma,l,d)=1-(1-p(\Sigma,l,d))^{m-l+1}$.
The probability that $M$ has at least one $d$-neighbor in each of $n$ random
strings of length $m$ is $q(m,\Sigma,l,d)^n$. Finally, the expected number of
spurious motifs in an instance with $n$ strings of length $m$ each is:
$|\Sigma|^l q(m,\Sigma,l,d)^n$. In this paper
we consider all combinations of $l$ and $d$ where $l$ is at most 50 and
the number of spurious motifs (expected by random chance) does not
exceed 500.
Note that for a fixed $d$, if we can solve instance $(l,d)$ we can also
solve all instances $(l',d)$ where $l'>l$, because they are less challenging
than $(l,d)$.

\section{Threshold where we switch from the sample to the pattern driven part}
\label{secThreshold}
Assume
that we switch to pattern generation as soon as the stack size is equal to $t$.
The number of $l$-mers in the first row is $m-l+1$. Assume that with each
$l$-mer we add to the stack, the number of surviving $l$-mers in each row
decreases with rate $p$. In other words, after we add one $l$-mer to the stack,
in each remaining row we are left with $p(m-l+1)$ $l$-mers. After we add $k$
$l$-mers to the stack we are left with $S_k=(m-l+1)p^k$ $l$-mers in each row.
The number of tuples (stacks) of size $k$ we expect to generate is then $T_k =
\Pi_{i=1}^{k-1} S_k = (m-l+1)^kp^{k(k-1)/2}$. For every tuple of size $k$ we
have to filter each of the surviving $l$-mers. There are $n-k$ rows to filter,
and each row contains $S_k$ items. Testing the filtering conditions for one
$l$-mer takes $O(l)$ time. Therefore, for each tuple of size $k$ filtering takes
$O(n S_k l)$ time. If we consider all the tuples of sizes up to $t$ we get the
following estimate on the runtime of the sample driven part: $Time_s(t) =
O(\Sigma_{k=1}^t n T_k l = nl \Sigma_{k=1}^t (m-l+1)^kp^{k(k-1)/2})$.

A simple upper bound for
$p$ can be obtained as follows. Define the number of $l$-mers at a distance of no more than $d$ from a given
$l$-mer as $N_d=\Sigma_{i=0}^d(_i^l)(|\Sigma|-1)^i$.  Consider only the $l$-mer
$u$ at the top of the stack. The probability that $u$ is at a distance no more than $2d$ from one of
the $l$-mers in the remaining rows of $R$ is $p=N_{2d}/\Sigma^l$. The
runtime above becomes $Time_s(t) =  nl \Sigma_{k=1}^t
(m-l+1)^kN_{2d}^{k(k-1)/2}/|\Sigma|^{lk(k-1)/2}$.

Next we look at the pattern driven part. We estimated that we generate $T_{t}$
tuples. For a tuple of size $1$, that is a single $l$-mer, there are
$N_d$ neighbors to enumerate. As $t$ increases, the number of $l$-mers we
enumerate per tuple decreases with rate $q$, and so, for a tuple of size $t$ we
enumerate $N_d q^{t-1}$ $l$-mers. Assuming perfect pruning, the time spent in
the pattern driven part is $Time_p(t) = O(T_t N_d q^{t-1} l)$.

To estimate $q$ consider the following. Say we know $M$ is a common neighbor for
a tuple of size $k$. Then we add one more $l$-mer $v$ to the tuple. We know that
$v$ is within distance $2d$ from the $l$-mer at the top of the stack. The 
probability that $M$ is within distance $d$ of $v$ is upper bounded
by $N_d/N_2d$. Here we make the generous assumption that the entire
$d$-neighborhood of $M$ is included in the $2d$-neighborhood of $v$. So an
upper bound for $q$ is $N_d/N_{2d}$ and $Time_p(t) =
O((m-l+1)^tN_{2d}^{t(t-1)/2}/|\Sigma|^{lt(t-1)/2} N_d^t/ N_{2d}^{t-1} l)$.

If we make the two times equal we unfortunately don't get a closed form
solution for the optimal $t$. However, the expressions can be computed for
every value of $t$ from 1 to $n$. We then can pick the best $t$. In practice we
use a much simpler formula where $t$ increases with $d$ and with $\log
|\Sigma|$ and decreases with $m$. The threshold
increases as $d$ gets bigger or $|\Sigma|$ gets bigger to avoid generating
very large neighborhoods and decreases with $m$ to avoid spending too much time
doing filtering. 

\end{document}